\documentclass[12pt,reqno]{amsart}

\usepackage{amsmath,amssymb}
\usepackage[all]{xy}

\makeatletter
\@addtoreset{equation}{section}
\makeatother

\usepackage{enumerate}

\def\ii{\sqrt{-1}}
\def\ee{\mathrm e}

\def\SD{\mathrm {SD}}

\def\calE{{\mathcal{E}}}
\def\calW{{\mathcal{W}}}
\def\CC{{\mathbb C}}
\def\cC{{\mathcal{C}}}
\def\rv#1{{{v^{(\mathrm{r})}_{#1}}}}
\def\iv#1{{{v^{(\mathrm{i})}_{#1}}}}
\def\rvp#1#2{{{v^{#1(\mathrm{r})}_{#2}}}}
\def\ivp#1#2{{{v^{#1(\mathrm{i})}_{#2}}}}
\def\by#1{{\rm{#1}, }}
\def\OmegaI{{{\Omega^{(\mathrm{I})}}}}
\def\OmegaII{{{\Omega^{(\mathrm{II})}}}}
\def\ZZ{{\mathbb Z}}

\def\RR{{\mathbb R}}
\def\PP{{\mathbb P}}

\def\idiff{{\mathfrak{idiff}}}
\def\iidiff{{\mathfrak{iidiff}}}
\def\IDiff{{\mathrm{IDiff}}}
\def\IIDiff{{\mathrm{IIDiff}}}

\def\bM{{\mathbb{M}}}
\def\fM{{\mathfrak{M}}}
\def\cM{{\mathcal{M}}}

\def\UU{\mathrm{U}}

\def\tw{{\tilde{t}}}

\newtheorem{theorem}{Theorem}[section]
\newtheorem{definition}[theorem]{Definition}

\newtheorem{proposition}[theorem]{Proposition}
\newtheorem{corollary}[theorem]{Corollary}
\newtheorem{remark}[theorem]{Remark}
\newtheorem{lemma}[theorem]{Lemma}

\def\AS#1{{{\mathcal A}_{S^1}^{#1}}}
\def\hAS#1{{\widehat{\mathcal A}_{S^1}^{#1}}}

\def\dfrac#1#2{{\displaystyle\frac{#1}{#2}}}

\def\paper#1{\textit{#1}, }
\def\jour#1{\rm{#1}, }
\def\yr#1{({\rm{#1}) }}
\def\vol#1{\textbf{#1}}
\def\pages#1{\rm{#1}}

\def\by#1{{\rm{#1}, }}

\begin{document}

\title{From Euler's elastica to the mKdV hierarchy, through the Faber
  polynomials }

\author{Shigeki Matsutani}
%\address{Industrial Mathematics, 
%National Institute of Technology, Sasebo College,\\
%1-1 Okishin, Sasebo, Nagasaki, 857-119, Japan\\
%smatsu@sasebo.ac.jp}
%\address{
%Institute of Mathematics for Industry, 
%Kyushu University\\
%744 Motooka, Nishi-ku, Fukuoka 819-0395, Japan
%}

\author{Emma Previato}
%\address{
%Department of Mathematics and Statistics,\\
%Boston University,\\
%Boston, MA 02215-2411, U.S.A.\\
%ep@bu.edu}

\bigskip

\maketitle

\begin{abstract}
The modified Korteweg-de Vries hierarchy (mKdV) is derived by imposing
isometry 
and isoenergy conditions on a moduli space of plane loops. The conditions are
compared to the constraints that define Euler's elastica. Moreover,
the conditions are shown to be
constraints on the curvature and other invariants of the loops which appear as
coefficients of the generating function for the Faber polynomials.
\end{abstract}

\section{ Introduction}
When the modified KdV (mKdV) equation was added to the 
unsystematically growing number of ``completely
integrable PDEs'', the \textit{ad hoc} explanation given was Miura's
 powerful idea: starting with the celebrated
KdV equation
$$u_t+6uu_x+u_{xxx},$$
``the simplest modification of the nonlinear term'' becomes mKdV; what is
now called  a ``Miura transformation'' \cite{miura}
 $u=v_x-v^2$ factors, with the  mKdV
equation as the right-hand factor,
$(-2v+\partial_x)(v_t-6v^2v_x+v_{xxx})=0$.
Solutions to mKdV thus transform into solutions of KdV, though the converse
does not always hold.

However, there 
appeared to be no direct motivation for mKdV, as opposed to the wave
motion in a shallow canal that was modeled by KdV in the nineteenth century,
even
though Miura (\textit{loc. cit.}) does refer to anharmonic lattices.
  Goldstein and Petrich\cite{GP1,GP2}
showed that going from the real $(x,y)$ to the complex $z=x+\ii y$
representation of a number on the plane, takes KdV to mKdV, and they obtained
the latter equation by the 
description of a plane region which moves in time with
conserved area and perimeter:  the
curvature satisfies the mKdV equation (up to rescaling).

The first-named author \textit{et al.}
in a series of papers\cite{Ma2, Ma4, Ma6, Ma7, Ma8}, defined
 ``quantized elastica'', based on Euler's
Elastica theory,
and derived mKdV.
Euler's Elastica have a long history, spanning from Pappus of
Alexandria around 300 A.D., to splines and computer vision
today, through the
discovery of the addition theorem for elliptic functions, which give the
closed-form solution, as thoroughly surveyed
 in a recent Ph.D. thesis\cite{levienTh, levien}, as well as
 Truesdell's classic works\cite{T1}
on the history of mechanics, particularly Euler's
works. 
Put in modern terms, the elastica are curves that solve an isoperimetric
variational problem, by minimizing a potential energy, but there are many
equivalent formulations; in particular, by writing the equivalent dynamics of
the simple swinging pendulum, the general solution can be parametrized by one
number, a Lagrange multiplier, representing the force of gravity in a suitable
normalization\cite{levien}. ``Quantized elastica'' replace the
Euler-Bernoulli energy functional\cite{T1}
(cf. Section \ref{setting} for the meaning of the symbols),
$$
\calE[Z]:=\oint \{Z, z\}_{\SD} ds = 
\frac{1}{2}\oint k^2 ds 
$$
with a `quantized' version
$$
{\mathcal Z}[\beta] = \int_{\fM} D Z \exp(-\beta {\mathcal{E}}[Z]),
$$ the partition function\cite{STY} of the temperature $\beta$.
The mKdV hierarchy then appears by deforming the coefficients of the expansion
of the integral at its extrema\cite{Ma2}.

In this paper, 
we derive the mKdV hierarchy by introducing time parameters
over a suitable  moduli space of plane loops; we use 
the setting  of Brylinski\cite{Br} to construct
differential flows on several related moduli spaces,
and we partition the appropriate moduli space
by the mKdV orbits under a commuting hierarchy of isometric and isoenergy
flows. 
We present the  invariant that defines the elastica as the integral
of a coefficient in the expansion
of a holomorphic function
used by Tjurin\cite{Tj} to represent projective connections on a Riemann
surface, so that the setting is that of the moduli space of curves;
both the isometric and isoenergy conditions 
 are therefore constraints on the coefficients of said expansion, which 
points out
 a relationship with 
modular forms and Faber polynomials (Subsection \ref{faber}).

\bigskip

{\bf{Acknowledgments:}} 
The first-named author expresses his thanks to Y. Kodama, V. Enolskii, 
M. Guest and A. Fujioka for their invitations to their seminars and 
discussions on this subject.
Both authors are very grateful to John McKay for supporting their visit in
Concordia University, where they had the opportunity to hear his talks and
questions about replicable functions and Faber polynomials. 

\bigskip

\section{Setting and review}\label{setting}
We work in the free loop space $LM$ of smooth (i.e., real analytic)
maps $S^1\rightarrow M$
(cf. Section 3.1 in Ref. \cite{Br}), where the (smooth, paracompact, oriented) manifold
$M$
is  the real  plane. In other words, our space $LM$ is the set of
real-analytic immersions of $S^1$ in $\mathbb{C}$, $Z: S^1 \hookrightarrow
\CC, 
Z(s)\in\cC^\omega(S^1, \CC) $,
which we parametrize by arclength; moreover, we fix the length
and take the quotient space under Euclidean motion:
$$
 \cM := \{Z: S^1 \hookrightarrow \CC | 
 Z(s)\in\cC^\omega(S^1, \CC), |d Z/ds|=1,
\ \oint |d Z| = 2\pi  \}, 
\quad \bM:=\cM/\sim.
$$
The condition $|d Z/ ds| =1$, namely
$d Z/ ds \in \UU(1)$,  allows us to define an action 
$
g_{s_0}Z(s) = Z(s-s_0)$ (where we identified $S^1$ with 
$\mathbb{R}$ modulo the integers); we call ``Euclidean motion''
a transformation with this
property.  We set:
$$
 \fM:=\bM/\UU(1).
$$
We denote the natural projections by
 $\mathrm{pr}_1: \cM \to \bM$ and $\mathrm{pr}_2: \bM \to \fM$.
For $Z\in\cM$, which we identify with $\mathrm{pr}_1(Z)$, 
$$
      \partial_s Z(s) =  \ee^{\ii \varphi(s)}
$$
where $\varphi$ is a real-analytic function.
The curvature of the loop is given by:
$$
 k(s_0):=\frac{1}{\ii}\frac{Z''(s_0)}{Z'(s_0)} = \partial_s\varphi(s_0), 
$$
so that, by equating separately the real and imaginary part\cite{GP1},  the
Fr\'enet-Serret relation can be expressed as, 
$$
 \partial_s (\partial_s Z) = k \ii (\partial_s Z).
$$
Moreover, 
the integral of $k$ with respect to $ds$
$$
\calW[Z]:=\frac{1}{\ii}\oint \frac{Z''(s)}{Z'(s)} ds = \oint k d s
$$
is the winding number of the loop.

The integral 
\begin{equation}
\calE[Z]:=\oint \{Z, s\}_{\SD} ds = 
\frac{1}{2}\oint k^2 ds, 
\label{eq:EBen}
\end{equation}
where the subscript SD denotes the Schwarzian derivative
(\ref{eq:SD}), 
is the Euler-Bernoulli energy functional--here we return to the representation
of $s$ as a complex number of norm 1.

To address applications to biology and other sciences, 
the first-named author\cite{Ma8}
 replaced the Euler-Bernoulli functional with:
$$
{\mathcal Z}[\beta] = \int_{\fM} D Z \exp(-\beta {\mathcal{E}}[Z])
$$
on a  measure space $(\fM, \mathfrak{F})$ with
measure $DZ$.

We note that these invariants appear in an expansion that 
was used by Tjurin\cite{Tj}
to encode the projective structure of a Riemann surface, but we point out that
in our case it is a formal expansion because, while $s$ is viewed as a complex
number by Tjurin (and in the expression of the Schwarzian derivative below),
it must be viewed as a real number when taking differences of the type
$s-s_0$.

\begin{equation}
\begin{split}
\frac{1}{2}\log\frac{Z(s_2) - Z(s_1)}{s_2-s_1}& 
=\frac{1}{2} \log Z'(s_0) 
+ \frac{1}{2}\frac{1}{2!}\frac{Z''(s_0)}{Z'(s_0)} [(s_1-s_0) +( s_2-s_0)]\\
& +\frac{1}{2} \frac{1}{3!} 
\left[ \frac{Z'''(s_0)}{Z'(s_0)} - \frac{3}{4}
\left( \frac{Z''(s_0)}{Z'(s_0)}\right)^2 \right][(s_2-s_0)^2 
      +(s_1-s_0)^2]\\
& +\frac{1}{2} \frac{1}{3!} \left[ \frac{Z'''(s_0)}{Z'(s_0)} - \frac{3}{2}
\left( \frac{Z''(s_0)}{Z'(s_0)}\right)^2 \right]( s_1-s_0)(s_2-s_0)\\
  + \cdots.
\end{split}
\label{eq:expansion}
\end{equation}
Using the representation 
$$
      \partial_s Z(s) =  \ee^{\ii \varphi(s)},
$$
the first term of the expansion is $\ii \varphi(s_0)/2$, the coefficient of
the second term is the curvature $k(s_0)$, the coefficients of the third and
fourth term are expressed by:
$$
\left[ \frac{Z'''(s_0)}{Z'(s_0)} - \frac{3}{4}
\left( \frac{Z''(s_0)}{Z'(s_0)}\right)^2 \right]
=-
\left[\frac{1}{4} k^2 - \ii \partial_s k\right]
$$

$$
\frac{1}{2} \left[ \frac{Z'''(s_0)}{Z'(s_0)} - \frac{3}{2}
\left( \frac{Z''(s_0)}{Z'(s_0)}\right)^2 \right]
=
\left[\frac{1}{4} k^2 + \ii \frac{1}{2}\partial_s k\right]
.
$$
We write them together:
$$
w^2 \pm \ii \partial_s w
=\left[\frac{1}{4} k^2 \pm \ii \frac{1}{2}\partial_s k\right]
$$
where $w=k/2$. The fourth coefficient is the Schwarzian derivative, namely the
operator,
\begin{equation}
\{f, z\}_{\SD}= 
 \left[ \frac{f'''(z)}{f'(z)} - \frac{3}{2}
\left( \frac{f''(z)}{f'(z)}\right)^2 \right],
\label{eq:SD}
\end{equation}
in our case
$$
\{Z, s\}_{\SD}= 
\left.
\frac{\partial}{\partial s_1}
\frac{\partial}{\partial s_2}
\frac{1}{2}\log\frac{Z(s_2) - Z(s_1)}{s_2-s_1} \right|_{s_2=s_1=s}.
$$

The integrals of the coefficients of both the third and fourth terms 
are therefore proportional to the Euler-Bernoulli energy
functional (\ref{eq:EBen}).

For an element
$Z \in \bM$,  $Z(0)$ vanishes up to Euclidean motion;
assuming this to be the case, we introduce the expansion,
\begin{gather*}
\begin{split}
	Z(s) &= s + a_1 s^2 + a_2 s^3 + a_3 s^4 + \cdots \\
	     &= \cdots + c_{-2} q^{-2} 
	     + c_{-1} q^{-1} + c_{1} q^{1} + c_{2} q^{2} 
             + c_{3} q^{3} + \cdots, \\
\end{split}
\end{gather*}
where $q = \ee^{\ii s}$ and will use the coefficients
 $a_i$ as moduli of the loops; we introduce a dependence of the $a_i$'s
on parameters $t=(t_2, t_3,  \ldots) \in [0,\varepsilon)$,
where $\varepsilon $ is a small positive number, to be viewed as deformations
independent of each other,
\begin{equation}
[\partial_{t_i}, \partial_{t_j}] = 0, \quad (i, j = 2, 3,  \ldots),
\label{eq:CRtt}
\end{equation}
 and we require that
 they also satisfy the
 condition 
\begin{equation}
[\partial_{s}, \partial_{t_i}] = 0, \quad (i = 2, 3,  \ldots),
\label{eq:CRst}
\end{equation}
which we call ``isometric''; we think of $s$ as the stationary flow $t_1$,
cf. Section \ref{translation}. 
We will derive the mKdV hierarchy by imposing an extra condition on the loops
$$
	Z = Z(s, t_2, t_3,  \ldots),
    \quad a_i = a_i(t_2, t_3,  \ldots).
$$

We denote by  $\AS{p}(\RR)$ the real-analytic $p$-forms on $S^1$ with
$\mathbb{R}$-valued coefficients, and by  
 $\AS{p}(\CC)$ the real-analytic $p$-forms on $S^1$ with
$\mathbb{C}$-valued coefficients.

\begin{lemma}\label{lm:oint}
The kernel of the $\RR$-linear map $\oint : \AS1(\RR) \to  \RR$ is
equal  to $d \AS0(\RR)$.
\end{lemma}

\begin{proof}
For $\eta \in \AS{1}(\RR)$, there exists $F\in \cC^\omega(\RR, \RR)$
such that $\eta := d F$ since $S^1 = \RR/2 \pi \ZZ$.
When $\eta$ is in the kernel of $\oint$, $F(2\pi) = F(0)$, so 
$F$ belongs to $\AS0(\RR)$.
\end{proof} 

For a given $Z\in \bM$,
the tangential angle $\varphi \in \AS0(\RR)$
and the curvature $k \in \AS0(\RR)$ are
defined as 
$$
\varphi(s):= \varphi[Z](s) :=\frac{1}{\ii}\log \partial_s Z(s) \in \AS0(\RR),
$$
$$
k(s):=k[Z](s):=\partial_s \varphi[Z](s),
$$
consistent with 
the terminology at the beginning of this Section.
 Since $|\partial_s Z(s)|=1$, for  fixed parameters $t_i$, 
we view the curvature $k$ as 
an element of $\AS0(\RR)$ and  $k ds $ as an element of $\AS1(\RR)$.

\section{Isometric Deformations}
We consider the tangent space of the isometric deformations
in $\bM$,
namely a fiber space 
$T_Z \bM$ over the tangent space $T\bM$ of $\bM$ at $Z$ in $\bM$.
Regarding the parameter $s$ in $Z$ as fixed,
the fiber is given by the isometric deformations of $\bM$.
We let $\IDiff(\bM)$ be the group acting on $\bM$
that gives isometric deformations, and
 we consider the action of
its Lie algebra $\idiff(\bM)$ at a point $Z$ in $\bM$.
We identify the infinitesimal orbit of $\idiff(\bM)$, or
the infinitesimal deformations, and the elements of 
$T_Z\bM$ at $Z \in \bM$. We do not need to prove the existence of this,
say, analytic Lie group, because  it
will consists of the flows that we construct concretely.
We take the limit of the flow  parametrized by
 $t \in [0, \varepsilon)$ at $Z \in \bM$ 
as $\varepsilon$ approaches zero,
in order to investigate the fiber
$T_Z \bM$. As a result, this limit gives us tangent vectors and differential
operators belonging to the fiber, without a need to introduce a connection: 
for our purposes, we
only  need local calculations for given coordinates.

A deformation is defined by:
\begin{equation}
	\partial_{t} Z(s) = v(s) \partial_s Z(s), \quad 
        s \in S^1,
\label{eq:ptZ}
\end{equation}
where $v$ belongs to $\AS0(\CC)$ and
$$
v(s) = \rv{}(s) + \ii \iv{}(s), \quad
$$
where $\rv{}$ and $\iv{}$ denote real and imaginary parts respectively,
so that
$$
\rv{},\iv{} \in \AS0(\RR):=\cC^\omega(S^1, \RR).
$$

Recalling that the ``isometric'' condition means that $s,t$ are independent
variables, namely $\partial_{t} \partial_s Z = \partial_{s} \partial_{t} Z$,
we check that:

\begin{proposition}
For a point $Z \in \bM$,
the isometric condition
$\partial_{t} \partial_s Z = \partial_{s} \partial_{t} Z$
is reduced to 
$$
\partial_{t} \partial_s \varphi(s) = \partial_{s} \partial_{t} \varphi(s)
\quad s \in S^1,
$$
where $\varphi=\varphi[Z]$.
\end{proposition}
\begin{proof}
The statement follows  from the calculations, $\displaystyle{
\partial_t \frac{\partial_s^2 Z}{\partial_s Z}
=
 \frac{\partial_t    \partial_s^2 Z}{\partial_s Z}
-\frac{\partial_s^2 Z \partial_s\partial_t Z}{(\partial_s Z)^2}
}$
and
$\displaystyle{
\partial_s \frac{\partial_t\partial_s Z}{\partial_s Z}}$
$\displaystyle{
=
\frac{\partial_s\partial_t \partial_s Z}{\partial_s Z}
}$
$\displaystyle{
-\frac{\partial_s^2Z \partial_t\partial_s Z}{(\partial_s Z)^2}
}$.
\end{proof}

An isometric deformation of $\partial_s Z = \ee^{\ii\varphi}$ 
belongs to $\cC^\omega(S^1, \UU(1))$, and its derivative
 along $t \in [0, \varepsilon)$ is given by
$$
\partial_t \partial_s Z = \ii \partial_t \varphi(s) \partial_s Z. 
$$
By applying $\partial_s$ to both sides of (\ref{eq:ptZ}), 
we have
$$
\partial_s \partial_{t} Z = \partial_s (v \partial_s Z).
$$
$$
\ii(\partial_t\varphi) \partial_{s} Z
 = (\partial_s v+v \ii k) \partial_s Z.
$$
Since $\partial_t \varphi$ is real-valued, we
obtain two conditions by equating the real and imaginary parts:

\begin{proposition}
For a point $Z \in \bM$,
the isometric condition $[\partial_s, \partial_t]Z = 0$ 
is reduced to two conditions{\rm{\cite{GP1,GP2}}}:
\begin{eqnarray}
	&\partial_tk &= \partial_s
           (\partial_s \frac{1}{k} \partial_s + k) \rv{} 
= \partial_s(\partial_s +  k \partial_s^{-1} k) \iv{}, 
\label{eq:condI}\\
	&  k \iv{} &= \partial_s \rv{}. \label{eq:condII}
\end{eqnarray}
\end{proposition}

In the limit as $\varepsilon$ approaches zero, these conditions hold on the
fiber of the tangent space.
Therefore we can define the following length-preserving linear
transformation on the tangent space, which we identified with $\idiff(\bM)$,
$$
\partial_s
\begin{pmatrix}
\iv{} \\ \rv{} 
\end{pmatrix}
=
\begin{pmatrix}
0 & \partial_s\frac{1}{k}\partial_s \\
{k} & 0 \\
\end{pmatrix}
\begin{pmatrix}
\iv{} \\ \rv{}. 
\end{pmatrix}, \quad
\begin{pmatrix}
\partial_s & -\partial_s\frac{1}{k}\partial_s \\
-{k} & \partial_s \\
\end{pmatrix}
\begin{pmatrix}
\iv{} \\ \rv{} 
\end{pmatrix}=0. \quad
$$

We derive some consequences from the second condition, (\ref{eq:condII}).
Let us consider the map,
$$
\ell_d : \AS0(\RR) \to \AS1(\RR), \quad
\ell_d(\iv{})=k\iv{}ds,
$$
and denote the inverse image of $d \AS0(\RR)\subset \AS1(\RR)$  by 
$\hAS0(\RR) := \ell_d^{-1}  d \AS0(\RR) $,
together with the  map,
$$
\widehat{\ell_d} : \hAS0(\RR) \to  d \AS0(\RR) \subset
\AS1(\RR), \quad
\widehat{\ell_d}(\iv{})=k\iv{}ds
=\partial_s \rv{}ds=d\rv{}.
$$
We define
$\ell_{r}: \hAS0(\RR) \to  \AS0(\RR)/\RR$ where
$
\ell_{r}(\iv{})=\int^s_0 k\iv{}ds =\int^s_0\partial_s \rv{}ds.
$
We will need (cf. Lemma \ref{lm:DtcDs}) 
only  deformations whose real part $\rv{}$ is
unique up to multiplicative constant, so
we fix a representative for the image in the quotient 
 $\ell_r^0 : \hAS0(\RR) \to  \AS0(\RR)$ of $\ell_r$ so that
$\ell_r^0(\iv{})=0$, by requiring 
\begin{gather}
\oint ds \ell_r^0(f) =0.
\label{eq:int0}
\end{gather}

We have an induced map,
$$
\ell: \hAS0(\RR) \to \AS0(\CC), \quad \ell(f) = \ell_r^0(f) + \ii f,
$$
which is a Euclidean motion,
$$
\partial_t Z |_{s=0} = \ell(\iv{} )\partial_s Z |_{s=0} = C.
$$
By introducing a loop of constant derivative,
$$
 \partial_t Z_0 = C,
$$
we can arrange for 
$$
\partial_t (Z-Z_0)|_{s=0} = C-C=0;
$$
similarly, we can add a constant to $\partial_t \phi(s)$ 
by using  rotation as the Euclidean motion.
We will use this choice of representative in $\bM$
for an element  $Z\in\cM$.

The following proposition is essentially 
the same as Proposition 3.3.4 (i) in Ref. \cite{Br}, 
 %\cite{Br},
following from (\ref{eq:condII}) for the map $\ell$.
\begin{proposition}
For $v\in \hAS0(\RR)$ and 
$\tilde{Z} \in \cM$, 
$\ell$ induces the bijection
$\ell^\sharp$ 
and the surjection 
$\ell^\flat$,
$$
\xymatrix{
 \hAS0(\RR) \ar[dr]^{\ell^\flat} \ar[r]^{\ell^\sharp}  
     & T_{\tilde Z}(\cM) \ar[d]^{\mathrm{pr}_{1*}}\\
 &  T_{\mathrm{pr}_1(\tilde Z)} (\bM) 
}
$$
in the sense that, for every element $v$ in the vector space $T_Z(\bM)$,
where $Z$ has the representative $\tilde{Z}$ in $\cM$,
we have an element $\iv{}\in \hAS0(\RR)$ such that
$v=\ell(\iv{})$.
\end{proposition}

Given this result, it will be more convenient and sufficient for our purposes
 to deal with  elements of $\bM$ rather than $\cM$.

\bigskip
 For the first condition,  (\ref{eq:condI}),
let us introduce the differential operators,
$$
	\OmegaI := \partial_s (\partial_s \frac{1}{k} \partial_s + k),
        \quad
	\OmegaII := \partial_s^2 + \partial_s (k \partial_s^{-1} k),
$$
where $\OmegaII $ is the recursion operator   for the mKdV hierarchy.

\begin{lemma}
The infinitesimal isometric deformation (\ref{eq:ptZ}) is reduced to
$$
	\partial_{t} k =  -\OmegaII \iv{}.
$$
as an element in the fiber of $T_Z \bM$.
\end{lemma}

\section{Isoenergy deformations}
We will now impose, in addition to the isometric condition, that
the Euler-Bernoulli energy functional $E$ is preserved by the deformation.
We call the Lie group of 
isometric and isoenergy motions  acting on $\bM$
$\IIDiff(\bM)\subset \IDiff(\bM)$, and
its Lie algebra  acting at a point $Z$ in $\bM$, $\iidiff(\bM)$.
We identify the infinitesimal orbit of an element of $\iidiff(\bM)$ with
the infinitesimal deformation itself 
and with the corresponding element of $T_Z\bM$ of $Z \in \bM$.

\begin{definition}
We define
$$
\bM_E := \{ Z \in \bM\ |\ \calE[Z] = E \mbox{ is preserved} \},
$$
and denote by $\mathrm{pr}_E$ the natural projection
$\bM_E \to \fM_E:=\bM_E/\UU(1)$.
\end{definition}

 From Lemma \ref{lm:oint}, we have the following result:

\begin{proposition}
 For $Z\in \bM$,
$\partial_t \calE(Z)$ vanishes iff $k\partial_t k ds \in d \AS0(\RR)$,
i.e.,  there exists a function $f \in \AS0(\RR)$ satisfying
$k \partial_t k = \partial_s f$.
\end{proposition}

Note that (\ref{eq:condII}), one of the isometric conditions, says that
 $k \partial_t k $ is an exact differential, akin to the isoenergy condition.

We now construct an isometric deformation: let   $\ivp{\prime}{}:=\partial_t k$ and  
$\rvp{\prime}{}:=\ell_r(\ivp{\prime}{})$, which satisfies 
$k \partial_t k = k \ivp{\prime}{} =\partial_s \rvp{\prime}{}$.
This yields the isometric condition,
$$
\partial_{t'} Z = (\ell(\ivp{\prime}{}))\partial_s Z,
$$
or
$$
\partial_{t'} k = \OmegaI\ivp{\prime}{}, 
\quad
k\ivp{\prime}{}=\partial_s \rvp{\prime}{}.
$$

The following proposition and corollary use this construction:

\begin{proposition}\label{prop:iidiff}
If two isometric deformations, $\iv{}, \ivp{\prime}{}\in \hAS0(\RR)$,
\begin{gather*}
\begin{split}
\partial_{t} Z  &= (\ell(\iv{}))\partial_s Z \quad 
\mbox{or}\quad
\partial_t k = \OmegaII \iv{},
\mbox{and}\\
\partial_{t'} Z  &= (\ell(\ivp{\prime}{}))\partial_s Z \quad
\mbox{or}\quad
\partial_{t'}k = \OmegaII \ivp{\prime}{},\\
\end{split}
\end{gather*}
satisfy the relation $\partial_t k = \ivp{\prime}{}$, then:
\begin{enumerate}

\item
the deformation $\partial_{t} Z$ is isoenergy, i.e.,
$ \partial_t \calE[Z] = 0$, and

\item 
$\partial_{t'}k = \OmegaII \ivp{\prime}{} = \OmegaII^2 \iv{}$.

\end{enumerate}
\end{proposition}

The converse also holds:
\begin{corollary}
If the isometric deformation, 
\begin{gather*}
\partial_{t} Z  = (\ell(\iv{}))\partial_s Z \quad 
\mbox{for which}\quad
\partial_t k = \OmegaII \iv{},
\end{gather*}
is isoenergy, then there is another isometric flow given by
$\ell(\ivp{\prime}{})$ satisfying $\partial_t k = \ivp{\prime}{}$.
\end{corollary}

\begin{remark}\label{rmk:rec0}
{\rm{
Proposition \ref{prop:iidiff} (2) gives rise to a recursive construction:
indeed, if $k\partial_{t'} k$ belongs to $d\AS0(\RR)$,
there is another isometric deformation
$\partial_{t''} k = \OmegaII \ivp{\prime\prime}{}$,
satisfying $\partial_{t'} k = \ivp{\prime\prime}{}$;
then the deformation $\partial_{t'} Z$ is also isoenergy,
and 
$$
\partial_{t''} k = \OmegaII^3 \iv{}
$$
holds.
}}
\end{remark}

It follows from this remark that,
\begin{proposition}\label{prop:rec0}
If $\iv{}\in \AS0(\RR)$ is given such that
each element of the sequence 
$\{\OmegaII^n \iv{}\}_{n=0,1,2,\ldots}$ belongs to ${\hAS0}(\RR)$,
by introducing parameters
$(\tw_1, \tw_2, \ldots) \in [0, \varepsilon)$,
%$\ell(\iv{})$ preserves the induced metric and the energy,
%we  have a sequence 
we obtain  infinitesimal isometric and isoenergy deformations
satisfying
$$
    \partial_{\tw_1} k
    = \OmegaII \iv{},
$$
$$
    \partial_{\tw_2} k = \OmegaII 
    \partial_{\tw_1} k
    = \OmegaII^2 \iv{},
$$
$$
    \partial_{\tw_3} k = \OmegaII \partial_{\tw_2} k
    = \OmegaII^2 \partial_{\tw_1} k
    = \OmegaII^3 \iv{},
$$
$$
    \partial_{\tw_4} k = \OmegaII \partial_{\tw_3} k
    = \OmegaII^2 \partial_{\tw_2} k
    = \OmegaII^3 \partial_{\tw_1} k
    = \OmegaII^4 \iv{},
$$
$$
{\large{\vdots}}
$$
\end{proposition}

However, the problem remains of finding 
an  initial isometric and isoenergy deformation
such that
each element of the sequence 
$\{\OmegaII^n \iv{}\}_{n=0,1,2,\ldots}$ belongs to ${\hAS0}(\RR)$. We will
construct a natural one in the next 
Section.

\section{The stationary deformation}\label{translation}

Since $Z(s)$ and $Z(s-s_0)$ have the same image loop, we call this the
``stationary'' deformation, in accordance with the terminology for the mKdV
equation, when the solution $u(x,t_2,\ldots , t_n,\ldots )$ is viewed as a
function of $x$ only (N.B. In previous work by the 
first-named author (Section 3 in Ref. \cite{Ma2}), 
this is called the ``trivial''
deformation). 

\begin{lemma} \label{lm:DtcDs}
 For any given real number $c \in \RR$ and $Z \in \bM$,
 $Z(s+c t)$,  the stationary deformation, is a solution of
$$
\partial_t Z = c Z.
$$
\end{lemma}

\begin{proposition}
The stationary deformation is compatible with  the 
 map $\mathrm{pr}_2: \bM \to \fM$
and with its lifting 
map $\mathrm{id}:\fM_E \to \fM_E$, the identity.
\end{proposition}

\begin{remark}\label{rmk:Main}
{\rm{
The stationary deformation is isometric and isoenergy,
according to the equation,
$$
\partial_{t_1} k = \OmegaII^0 \partial_s k =  \partial_s k.
$$
}}
\end{remark}

\begin{proposition}
 For $Z\in \bM$ and $k:=k[Z]$,
starting with the stationary deformation:
$$
     \partial_{t_1} k = \partial_{s} k,
$$
each element of the sequence 
$\{\OmegaII^n \partial_s k\}_{n=0,1,2,\ldots}$ belongs to ${\hAS0}(\RR)$,
yielding 
a sequence of isometric and isoenergy
relations,
$$
    \partial_{t_2} k = \OmegaII \partial_{t_1} k
    = \OmegaII \partial_{s} k,
$$
$$
    \partial_{t_3} k = \OmegaII \partial_{t_2} k
    = \OmegaII^2 \partial_{t_1} k
    = \OmegaII^2 \partial_{s} k,
$$
$$
    \partial_{t_4} k = \OmegaII \partial_{t_2} k
    = \OmegaII^2 \partial_{t_2} k
    = \OmegaII^3 \partial_{t_1} k
    = \OmegaII^3 \partial_{s} k,
$$
$$
{\large{\vdots}}
$$
This recursive sequence is known as the modified KdV hierarchy.
\end{proposition}

\begin{proof}
We let $\iv{}$ in Proposition \ref{prop:rec0} correspond to the 
stationary deformation $\iv{}=\partial_s k$,
and $\tw_{i}=t_{i-1}$ for $i=2, 3, \ldots$.
The equalities hold because the sequence of $t_i$ are the higher mKdV flows,
since $\OmegaII$ is the mKdV recursion operator.
%at the $i$-th stage 
% because the $\partial_{t_i} k$
%are Hamiltonians 
%of the mKdV hierarchy, and the ones preceding the $i$-th are
%held constant, being conserved quantities\cite[Subsection 3.1]{Ma2}.
\end{proof}

\begin{remark}\label{rmk:Main2}
{\rm{
 For the stationary deformation, we can write:
$$
\partial_s k = \OmegaII 0
=(\partial_s^2 + \partial_s k \partial_s^{-1} k) \cdot 0=
\partial_s k \partial_s^{-1} k\cdot 0
$$
by choosing the factor, $\partial_s^{-1} 0=1$,
and thus write,
$$
\partial_{t_1} k = \OmegaII 0. 
$$
Then we can take $\iv{}$ in Proposition \ref{prop:rec0} 
equal to  $0$. Then, with
$\iv{} = 0$,
$\tw_{i}=t_{i}$ for $i=1, 2, 3, \ldots$.
}}
\end{remark}
By letting $w_j = \partial_{t_j} \varphi$, we also have
a matrix format:
\begin{gather*}
\partial_s
\begin{pmatrix}
1 \\ w_1 \\ w_2 \\ \vdots \\ w_{\ell} \\  \vdots
\end{pmatrix}
=
\begin{pmatrix}
          &          &          &          &          &   & \\
\OmegaII &  &          &          &          &                  &   \\
          & \OmegaII & &         &           &                   &  \\
          &          &          & \ddots  &           &         &    \\
          &          &          &         & \OmegaII &           &   \\
          &          &          &         &  &\ddots           &   \\
\end{pmatrix}
\partial_s
\begin{pmatrix}
1 \\ w_1 \\ w_2 \\ \vdots \\ w_{\ell} \\  \vdots
\end{pmatrix}.
\end{gather*}

It is well-known that the mKdV hierarchy is integrable
and its orbits are well-defined, thus the Lie-group and Lie-algebra
actions we used, for isometric and isoenergy
diffeomerphism $\IIDiff(\bM)$ on $\bM$,  make sense.

If the curvature $k$ is given, 
$\varphi$ is determined modulo a constant of integration.
If $\varphi$ is given, $Z$ is determined modulo Euclidean motion.
Therefore, the mKdV hierarchy acts on the space $\fM$
rather than $\bM$:

\begin{proposition}
There exists an action $\IIDiff(\fM)$ on $\fM$ given by the
mKdV hierarchy.
\end{proposition}

\section{Related constructions, finite orbits, Comments}

In this section, we  review 
previous results (Prop. 3.11 in Ref.\cite {Ma6}), to relate the KdV flow
with the 
deformations we described; we 
give equations for a class of finite mKdV orbits; and conclude with
comments. 

\subsection{Connection between KdV flow and mKdV flow}

In this subsection, we show the connection between the KdV flow 
(Prop. 3.11 in Ref.\cite {Ma6})
and the mKdV flow described above. 

For a solution $\psi$ of
$$ 
\left(-\partial_s^2-\frac{1}{2}\{Z, s\}_{SD}\right) \psi=0,
$$ 
the deformation $\partial_t Z = v \partial_s Z$ induces the
deformation,
$$
\partial_t \psi = -\frac{1}{2}(\partial_s v) \psi + v \partial_s \psi 
$$
by direct computation.
When $\iv{}= \partial_s k$ is the case, $v = \frac{1}{2} k^2+\ii
\partial_s k$. Since this deformation gives the mKdV hierarchy,
$v$ induces the KdV hierarchy, via the Miura map.
 
We have the natural symplectic structure on $T\mathfrak{M}$,
\begin{gather}
\langle Y_1, Y_2\rangle
=\oint_{S^1}\left(
\int^s (
Y_2(s) Y_1(s') -
Y_1(s) Y_2(s') ) ds'\right) ds,
\label{eq:sympKdV}
\end{gather}
under which   the KdV flow is Hamiltonian.
Moreover, the KdV hierarchy has a bi-Hamiltonian structure,
one of the two compatible structures being
associated to the Hamiltonian flow which 
 preserves the Euler-Bernoulli energy of the elastica.

%}}
%\end{remark}

\subsection{Symplectic structures}

Besides the symplectic structure for the mKdV flow
associated with  (\ref{eq:sympKdV})
through the correspondence between KdV  and mKdV,
 we have a sequence of natural symplectic structures, as follows.

For two deformations
$\partial_t Z = u \partial_s Z$ and 
$\partial_{t'} Z = u \partial_s Z$,
($v(s) = \rv{}(s) + \ii \iv{}(s)$, $u(s) = u^\mathrm{r}{}(s) + \ii
u^\mathrm{i}{}(s)$), 
we define:
\begin{gather*}
\begin{split}
\langle u,v\rangle_\ell
&=\frac{1}{2}\int  \mathfrak{Im}
\left(\overline{\partial_{t'}Z}\ 
\partial_{t}Z\right)\  k^\ell \ ds \\
&=\frac{1}{2}\int 
\left(
u^\mathrm{r}{}(s) \iv{}(s)-
u^\mathrm{i}{}(s) \rv{}(s)\right)\  k^\ell(s)\  ds,
\end{split}
\end{gather*}
where $\ell$ is a natural number.
Clearly, 
$\langle u,v\rangle_\ell
=-\langle v,u\rangle_\ell$. Moreover,
for $u$ and $v$ giving isometric deformations,  
$$
\langle u,v\rangle_\ell
=\int k^{\ell-1}(s) 
\left(\rv{}(s) \partial_s u^\mathrm{r}{}(s)\right)\  ds
$$
by straightforward computation.
The non-degeneracy is also clear:
for a point $u$ of $T{\mathfrak{M}}$ such that
$\langle u,v\rangle_\ell =0$ for every $v$ of $T{\mathfrak{M}}$,
 $u = 0$.
For the $\ell=1$ case, this gives the natural symplectic structure
of $T{\mathfrak{M}}$.

\subsection{Finite-dimensional orbits of a point $Z$}

The (disjoint)
orbits of the action $\IIDiff(\fM)$ on $\fM$ give a decomposition of
$\fM$.

The orbits are finite-dimensional when the hierarchy defined recursively
contains a finite number of linearly-independent flows. 
This happens, in particular, 
when the deformations $\partial_{t_i}$ are stationary for $i>g+1$,
although this is not the only case: as an analogy,
 the flows of a finite-dimensional orbit of the KdV equation
span the Jacobian of a spectral curve, of  given genus $g$; they
are coordinatized by the coefficients of the expansion of a basis of
holomorphic differentials in (the inverse of) a local parameter $t$
on the curve,
$\omega_i=\sum_{j=0}^\infty c_j^it^j$, $i=1,...,g$,
and of course all vectors $(c_j^1,...,c_j^g)$ could be non zero.
We call the union of the particular finite-dimension orbits
for which the flows become stationary exactly at the $g$ stage $\bM_g$,
although the orbit in fact could have any smaller dimension,
equal to the dimension of the span of the first $g-1$ flows:
$$
\bM_g:=\{Z \in \bM \ | \ \mbox{for } k =k[Z], 
   \OmegaII^g k =0\},
$$
$$
\fM_g:=\{Z \in \fM \ | \ \mbox{for } k =k[Z],
   \OmegaII^g k =0\} = \mathrm{pr}_2 (\bM_g),
$$
giving a filtration:
$$
\bM_g \subset \bM_{g+1}, 
\quad\mbox{and}\quad
\fM_g \subset \fM_{g+1}.
$$

%\begin{proposition}
%$\fM$ is decomposed by the orbits of mKdV hierarchy, and
%$$
%\bM = \bigcup_{g=1} \bM_g,
%\quad\mbox{and}\quad
%\fM = \bigcup_{g=1} \fM_g.
%$$
%\end{proposition}

The $g=1$, non-stationary case  corresponds to Euler's elastica\cite{Ma8}.

\begin{proposition}
The elements of the subspace $\bM_g$ are defined by the condition,
\begin{gather*}
\partial_s
\begin{pmatrix}
1 \\ w_1 \\ w_2 \\ \vdots \\ w_{g-1} \\ w_g 
\end{pmatrix}
=
\begin{pmatrix}
          &          &          &          &          &  \OmegaII \\
\OmegaII &  &          &          &          &                    \\
          & \OmegaII & &         &           &                    \\
          &          &          & \ddots  &           &                    \\
          &          &          &         & \OmegaII &             \\
\end{pmatrix}
\partial_s
\begin{pmatrix}
1 \\ w_1 \\ w_2 \\ \vdots \\ w_{g-1} \\ w_g 
\end{pmatrix},
\end{gather*}
where the multiplicative-constant ambiguity of  each $\OmegaII$ is 
adjusted by rescaling the corresponding  $t_j$.
\end{proposition}

\subsection{Action on cohomology}

As in Ref. \cite{Ma6},
we can interpret the hierarchy of flows on the filtration by orbits, by
viewing $\bM$, the loops in the Euclidean plane,  as the topological space obtained by
compactifying $\CC$ to $\PP$. We recall the result
 (Section III.16 in  Ref. \cite{BT}):

\begin{theorem}
The cohomology of the loop space $\Omega S^n$ over $S^n$ is given by
$$
\mathrm{H}^p(\Omega S^n, \RR) = \RR \delta_{[p \ \mathrm{mod} (n-1)], 0}. 
$$
As for the ring structure, write:
$$
\mathrm{H}^*(\Omega S^n, \RR) = 
\RR + \RR x + \RR e + \RR xe +
 \RR e^2 + \RR xe^2 +\cdots;
$$
then, for the $n=2$ case, the ring structure is given by
$$
\mathrm{H}^*(\Omega S^2, \RR) = \RR[x]/(x^2)\cdot\RR[e],
$$
where degree$(e) = 2$ and degree$(x) = 1$.
\end{theorem}

On the other hand, in Prop. 7.1 of
Ref.\cite {Ma6}, the following is proved,
\begin{theorem}
For the forgetful functor from the category of differential
geometry to that of topological spaces,
$\mathrm{F} : \mathit{Diff} \to \mathit{Top}$, we have
$$
\mathrm{H}^*(\Omega S^2, \RR) = 
\mathrm{H}^*(\mathrm{F}(\bM), \RR)  
$$
i.e., for
$\mathrm{H}^*(\Omega S^2, \RR) = 
\RR[x]/(x^2)\cdot\RR[e]= 
\mathrm{H}^*(\mathrm{F}(\bM), \RR)  
= \Lambda_\RR[dt_1,\epsilon]$,
where 
$\Lambda_\RR[dt_1,\epsilon]$ is a ring generated by
 $dt_1$ and 
$$
\epsilon = dt_1 + dt_2 \wedge (dt_1 i_{\partial_1}) +
 dt_3 \wedge (dt_1 i_{\partial_1}) + \cdots
$$
with the  wedge product and the degree$\mathrm{:}$ degree$(dt_i)=1$,
$$
\mathrm{H}^*(\mathrm{F}(\bM), \RR) = 
\RR + \RR dt_1 + \RR \epsilon + \RR  \epsilon dt_1
 + \RR \epsilon^2 + \RR  \epsilon^2 dt_1+\cdots.
$$
\end{theorem}

\begin{proof}
Since $\epsilon\cdot 1 = dt_1$, and 
$
\epsilon^{n-1}\cdot dt_1 = \epsilon^{n}\cdot 1 = 
dt_n \wedge dt_{n-1} \wedge \cdots \wedge dt_2\wedge dt_1$,
we have
\begin{gather*}
\begin{split}
\Lambda_\RR[dt_1,\epsilon]&=
%\mathrm{H}^*(\mathrm{for}(\bM), \RR) &= 
\RR + \RR dt_1 + \RR \epsilon + \RR  \epsilon dt_1
 + \RR \epsilon^2 + \RR  \epsilon^2 dt_1+\cdots\\
&=
\RR + \RR dt_1 + \RR dt_1\wedge dt_2 
+ \RR dt_1\wedge dt_2 \wedge dt_3 +\cdots.\\
\end{split}
\end{gather*}
%Due to the {\bf Darboux transformation},
The B\"acklund transformation acts on
$\bM$ as a telescopic-type 
topological space  according to the genera.
In conclusion: 
$$
\mathrm{H}^*(\mathrm{F}(\bM), \RR)
=\Lambda_\RR[dt_1,\epsilon].
$$
\end{proof}

\subsection{Comments}\label{faber}

The derivation of the mKdV hierarchy in this paper
is based on the conservation of geometric and energy properties
of ``soliton elastica''.
Moreover, the conserved quantities appear
in the expansion (\ref{eq:expansion})
of  the generating function $\log \dfrac{Z(s)-Z(s')}{s-s'}$ of
the Faber polynomials\cite{FMN, matsutanipreviato}.
Indeed, the Faber polynomials $P_{f,n}$ 
 for a function
$$
f(q) = \frac{1}{q} + a_1 q + a_2 q^2 + a_3 q^3 \cdots.
$$
are defined as
\begin{gather*}
\begin{split}
\log\left(q (f(q) - f(p))\right)
&= \log\left(1 - f(p) q+ a_1 q^2 + a_2 q^3 + a_3 q^4 \cdots
             \right)\\
&= -\sum_{n=1} \frac{1}{n} P_{f, n}(f(p)) q^n
\end{split}
\end{gather*}
and $P_{f, 0}(f(p))=1$, so that
\begin{gather*}
\begin{split}
\log
\left(pq\frac{f(p)- f(q)}{p-q}\right)
=\sum\frac{1}{n}\left(  P_{f, n}(f(p)) q^n
-\left(\frac{q}{p}\right)^n \right).
\end{split}
\end{gather*}

Notice that, setting $Z=Z(s), \ Z^\prime =Z(s^\prime ), \ p=1/Z,
q=1/Z^\prime,$ and calling the inverse functions $s=g(p),\ s^\prime =g(q)$,
then:
$$
\log\frac{Z(s)-Z(s^\prime )}{s-s^\prime}=-\log\left(pq\frac{g(p)-
    g(q)}{p-q}\right).
$$

This seems to be an important
connection with the problem of identifying and acting on the
``replicable functions'' $f(q)$, for which,
$$qp\frac{f(q)-f(p)}{q-p}=\exp\left(-\sum_{n,m\ge 1}h_{m,n}p^mq^n\right) ,
$$
where $h_{m,n}$ is the Grunsky coefficient:
\begin{gather*}
\sum_{m, n} h_{m,n} p^m q^n := \log
\left(pq\frac{f(p)- f(q)}{p-q}\right)= \log
\left(\frac{f(p)- f(q)}{\dfrac{1}{q}-\dfrac{1}{p}}\right),
\end{gather*}
and
$$\{ f, g\}_\mathrm{SD}=6\sum_{n,m\ge 1}mnh_{m,n}q^{m+n-2}.
$$

%\begin{thebibliography}{BBEIM}

\newpage
\noindent
Shigeki Matsutani

\smallskip

\noindent
Industrial Mathematics, 

\noindent
National Institute of Technology, Sasebo College,

\noindent
1-1 Okishin, Sasebo, Nagasaki, 857-119, Japan

\noindent
smatsu@sasebo.ac.jp
\smallskip

\noindent
Institute of Mathematics for Industry, 

\noindent
Kyushu University

\noindent
744 Motooka, Nishi-ku, Fukuoka 819-0395, Japan

\bigskip

\noindent
Emma Previato

\noindent
Department of Mathematics and Statistics,

\noindent
Boston University,

\noindent
Boston, MA 02215-2411, U.S.A.

\noindent
ep@bu.edu

\bigskip

\end{document}